\documentclass[final]{IEEEtran}
\usepackage{graphicx} 
\usepackage{footnote}
\usepackage{amsmath}
\usepackage{mathtools}
\usepackage{amssymb} 
\usepackage{enumerate}
\usepackage{array}
\usepackage{color}
\usepackage{gensymb}

\usepackage{subcaption}
\DeclareCaptionLabelSeparator{periodspace}{.\quad}
\usepackage{tikz}
\usetikzlibrary{calc,shapes,snakes,angles,positioning,intersections,quotes,decorations.markings}
\usepackage{tkz-euclide}
\usetkzobj{all}

\usepackage{pgfplots}
\pgfplotsset{compat=1.11}

\usepackage{verbatim}

\DeclarePairedDelimiterX\setc[2]{\{}{\}}{\,#1 \;\delimsize\vert\; #2\,} 

\makeatletter
\newsavebox{\mybox}\newsavebox{\mysim}
\newcommand{\distras}[1]{%
  \savebox{\mybox}{\hbox{\kern3pt$\scriptstyle#1$\kern3pt}}%
  \savebox{\mysim}{\hbox{$\sim$}}%
  \mathbin{\overset{#1}{\kern\z@\resizebox{\wd\mybox}{\ht\mysim}{$\sim$}}}%
}

\usepackage{amsthm}

\newtheorem{theorem}{Theorem}

\newtheorem{lemma}[theorem]{Lemma}

\theoremstyle{remark}
\newtheorem*{remark}{Remark}

\theoremstyle{definition}
\newtheorem{definition}{Definition}

\theoremstyle{definition}
\newtheorem{assumption}{Assumption}

\theoremstyle{definition}
\newtheorem{approximation}{Approximation}

\usepackage[cmintegrals]{newtxmath}

\usepackage{url}

\begin{document}


\title{A Mathematical Justification for Exponentially Distributed NLOS Bias \\[0.2ex]
\thanks{The authors are with Wireless@VT, Bradley Department of Electrical and Computer Engineering, Virginia Tech, Blacksburg, VA 24061 USA (e-mail: olone@vt.edu; hdhillon@vt.edu; buehrer@vt.edu).}
\thanks{The work of C. E. O'Lone was supported in part by the Bradley Graduate Fellowship through the Bradley Department of Electrical and Computer Engineering, Virginia Tech.  The work of H. S. Dhillon was supported by U.S. NSF Grant ECCS-1731711.}}
\author{Christopher~E.~O'Lone,~Harpreet~S.~Dhillon,~and~R.~Michael~Buehrer \vspace{-20pt}}

\maketitle


\begin{abstract}
In the past few decades, the localization literature has seen many models attempting to characterize the non-line-of-sight (NLOS) bias error commonly experienced in range measurements.  These models have either been based on specific measurement data or chosen due to attractive features of a particular distribution, yet to date, none have been backed by rigorous analysis.  Leveraging tools from stochastic geometry, this paper attempts to fill this void by providing the first analytical backing for an NLOS bias error model.  Using a Boolean model to statistically characterize the random locations, orientations, and sizes of reflectors, and assuming first-order (\emph{i.e.}, single-bounce) reflections, the distance traversed by the \emph{first-arriving} NLOS path is characterized. Under these assumptions, this analysis reveals that NLOS bias exhibits an exponential form and can in fact be well approximated by an exponential distribution -- a result consistent with previous NLOS bias error models in the literature.  This analytically derived distribution is then compared to a common exponential model from the literature, revealing this distribution to be a close match in some cases and a lower bound in others.  Lastly, the assumptions under which these results were derived suggest this model is aptly suited to characterize NLOS bias in 5G millimeter wave systems as well. \\[-3.5ex]
\end{abstract}



\vspace{-17pt}
\section{Introduction}

	Localization techniques utilizing range measurements, \emph{e.g.,} time-of-arrival (TOA) or time-difference-of-arrival (TDOA), commonly rely on LOS assumptions for best performance.  However, the ranging signal is frequently subjected to channel effects, such as reflections or diffraction, which leads to range measurements that are erroneously larger than the true distance between the transmitter and receiver.  Thus, in addition to noise, a positive bias term needs to be accounted for in the measurement.  The approach to dealing with this term has been the subject of much research over the years in the localization literature \cite[Sec. VII]{Char_NLOS_Bias}, \cite{NLOS_Survey}.

	While there are many aspects to the NLOS bias problem, such as algorithmic considerations \cite{NLOS_Survey} and localization performance modeling \cite{Qi_1}, one fundamental issue that arises is how to statistically characterize the NLOS bias.
In 2007, \cite{Mailaender} summarized the state of the NLOS bias research: ``At the present time, very little is known about the statistics of the NLOS variables [bias] in realistic propagation environments, and there are no established models.''  
Since then, there have been many proposed NLOS bias distributions such as a skew $t$ \cite{Skew_t}, a half-Gaussian \cite{Mailaender}, a Rayleigh \cite{GMix1}, a shifted Gaussian \cite{GMix1}, and a positive uniform \cite{DrB}.  Prior to these, a gamma distribution was also proposed \cite{Qi_2} and an exponential distribution was, and still is, a commonly used model \cite{Swaroop}, \cite{Chen}, \cite{exp_justification}.  There have also been previous attempts to use measurements to generate purely empirical NLOS bias distributions \cite{Win_Empirical}.  However, all of these bias distributions were chosen solely due to their tractability, desirable features, \emph{e.g.} a positive support, or their close fit to empirical measurements in specific scenarios. To date, none have been theoretically verified. 

	In an effort to finally bring an analytical backing to the NLOS bias problem, it is helpful to utilize tools from stochastic geometry; in particular, the Boolean model \cite{Stoyan}.  This model places distributions on the random positions, orientations, and sizes of buildings within a network.  The seminal work utilizing the Boolean model to study propagation within a cellular network was conducted in \cite{Heath}, in which the model was used to derive metrics such as connectivity and coverage probability.   
	
	The most relevant aspect of the Boolean model, in terms of addressing the NLOS bias problem, is its ability to handle reflections.  With the recent push towards 5G cellular and its use of millimeter wave (mm-wave) frequencies, diffraction effects become negligible while reflections begin to dominate propagation \cite{AndrewsHeath}.  Thus, the Boolean model becomes an attractive analytical tool.  In \cite{Nor}, the model was used to generate a power delay profile (PDP) under first-order reflections, independent blocking, and under the condition that all buildings have the same orientation per a channel realization.  
The work in \cite{GDasConf} and \cite{GDasJournal} extend that of \cite{Nor} by considering buildings with random orientations and equipping the transmitter and receiver with directional antennas.  
For this more general model, channel characteristics such as the PDP, the average number of reflections, etc., were derived under first-order reflections and independent blocking.  
Lastly, in an effort to determine coverage probability, \cite{Aroon} derives the distribution of the shortest reflected path. This was derived by considering a PPP of base stations, with the mobile at the origin, and considering the closest reflector to the mobile (with reflectors modeled as line segments).  Then, the distribution of the shortest path among the eligible base stations, through this single reflector, to the mobile was derived. 
	
	In terms of the NLOS bias problem however, the aforementioned works either don't offer the results we need or don't utilize the assumptions we desire in addressing the NLOS bias problem.  That is, while \cite{Nor} derives the PDP under first-order reflections, it does not provide the distribution of the first-arriving NLOS path length. \emph{Since the range measurement is determined via the first-arriving signal, it is desirable to have a characterization of only the first-arriving NLOS signal, not a characterization of all NLOS signals.}  Additionally, the model does not assume random orientations of buildings.

	Similarly, \cite{GDasConf} and \cite{GDasJournal} also do not derive a distribution for the shortest NLOS path length, and furthermore, the underlying model (of having the directional antennas at the transmitter and receiver pointed in fixed directions) makes it difficult to determine whether a shorter NLOS path may be available in another direction.  Lastly, although the work in \cite{Aroon} does address finding the shortest NLOS path, the model does not apply to the NLOS bias problem, since a range measurement is for a \emph{single} link, not for a mobile with multiple possible base station links.

\vspace{-11pt}	
\subsection{Contributions} \label{ContributionsSection}	
\vspace{-1pt}	
	The goal of this paper is to provide the first theoretical justification for a choice of an NLOS bias model.  This is achieved by considering a single transmitter-receiver pair, whose LOS path is assumed to be blocked, and which is surrounded by buildings (\emph{i.e.}, reflectors) distributed according to a Boolean model.  For this range measurement setup, we derive a distribution for the path length of the first-arriving NLOS signal under first-order reflections.\footnote{Since the NLOS path length is related to the NLOS bias via the constant distance, $d$, between the Tx and Rx, we refer to the NLOS `bias' and `path length' interchangeably, \emph{i.e.} path length = bias + $d$.}
	In achieving this goal, the paper makes the following contributions:
\begin{enumerate}
\vspace{-1pt}
\item The \emph{reflection region} is defined and derived.  This region characterizes all of the possible placements of a given reflector which produce a first-order reflection between the transmitter and receiver.  This region is a simple, elegant generalization of the `feasible region' in \cite{GDasConf} for omni-directional links, and is derived through a simple appeal to polar coordinates.
\item The distribution of the first-arriving NLOS path length is analytically derived and yields a form similar to an exponential.  This offers an analytical backing for the exponentially distributed NLOS bias models commonly assumed in the localization literature, \emph{e.g.} \cite{Swaroop}, \cite{Chen}, \cite{exp_justification}.  Additionally, we comment on the Boolean model properties which give rise to this exponential form.
\item To highlight the close connection with an exponential distribution, we show that this NLOS distribution can be approximated by a true exponential distribution.
\item Lastly, we compare the true analytical distribution of the first-arriving NLOS path length, along with its approximation, to that generated from a simulated link under a Boolean model, revealing complete agreement with the analytically derived distribution and a close agreement with the exponential approximation.
\end{enumerate}

\vspace{-10pt}
\section{Network Model} \label{Model_Assumptions}
\vspace{-1pt}
	This section outlines important definitions as well as our assumptions and gives a brief description of our setup. \emph{For the remainder of the paper, it is assumed we work in $\mathbb{R}^2$.} 
\vspace{-3pt}
\begin{definition} \label{MinkowskiSum}
\emph{(Minkowski Sum \cite[Ch. 1]{Stoyan})} Let $A$, $B \subset \mathbb{R}^2$ be compact. Then, the \emph{Minkowski sum} is defined as
\vspace{-3pt}
\begin{align*}
A \oplus B \triangleq \setc[\Big]{ x + y \in \mathbb{R}^2 } {x \in A,~ y \in B }.
\end{align*}
\end{definition}
\vspace{-8pt}

\begin{remark}
Intuitively, the Minkowski sum of $A$ with $B$ is an ``enlargement, translation, and deformation'' of the set $A$ \cite{Stoyan}.
\end{remark}

\vspace{-5pt}	
\begin{definition} \label{Boolean}
(\emph{Boolean Model \cite[Ch. 3]{Stoyan})} Let $\mathcal{R}$ be any set of \emph{compact}, non-empty objects in $\mathbb{R}^2$.  Then, a \emph{Boolean model}, $\mathcal{B} \subset \mathbb{R}^2$, is a random set with the following properties:
\vspace{-2pt}
\begin{itemize}
\item The placement of all of the object centers is defined by a \emph{homogeneous} PPP, $\Phi = \{\mathbf{c}_i\}_{i=1}^\infty$, where $\mathbf{c}_i = [x_i, y_i]^T$;
\item The sequence of objects, $\{R_i\}_{i=1}^\infty$, to be placed at the points in $\Phi$ are identically sampled from $\mathcal{R}$, independently of each other and of $\Phi$.
\end{itemize}

\begin{figure}[t]
\centering
\includegraphics [scale=0.32]{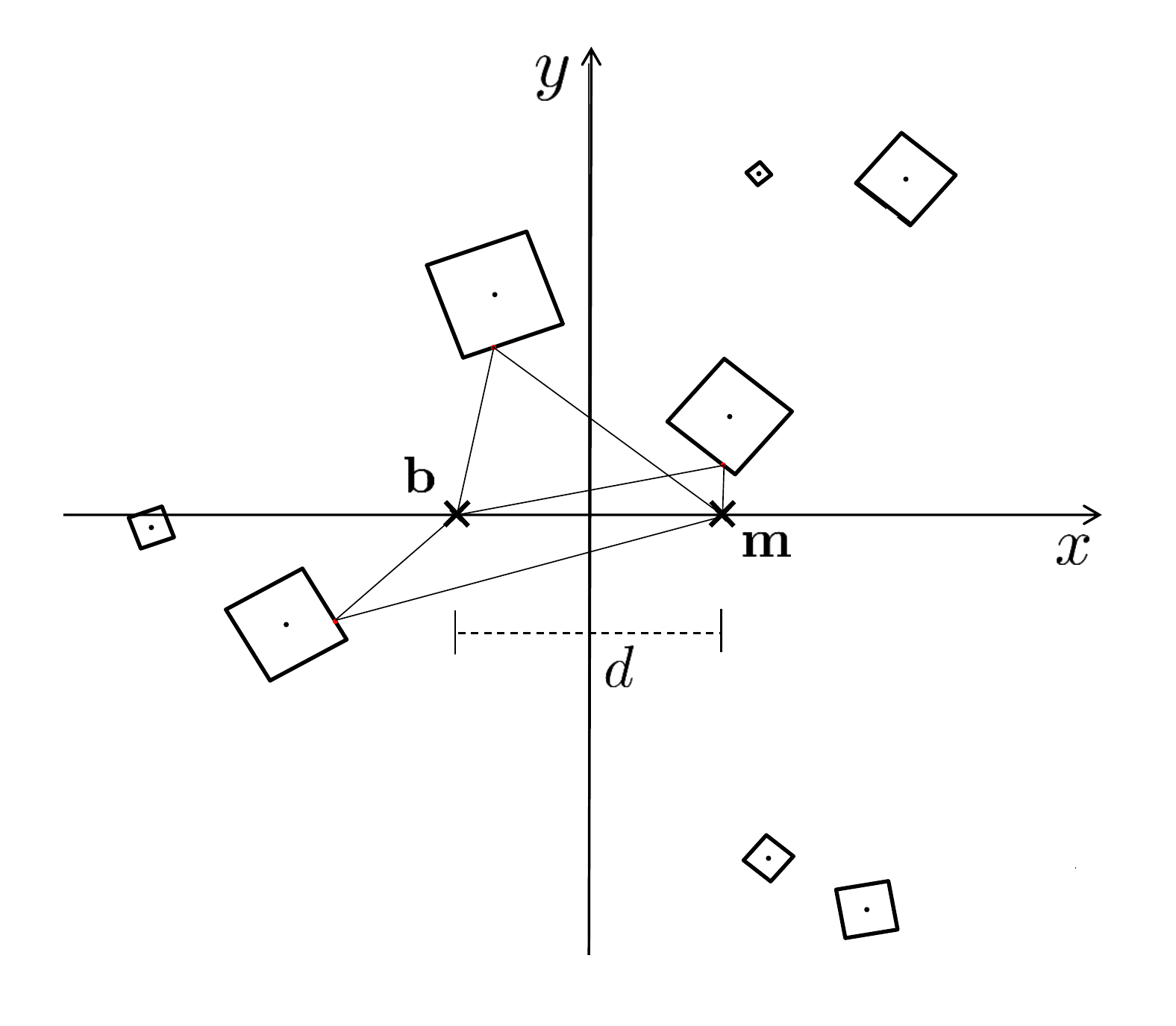}
\vspace{-19pt}
\caption{\textsc{A Boolean Model Realization Over the Test Link Setup}.\\[-6.0ex]} 
\label{BM_Realization}
\end{figure}

\vspace{-2pt}
\noindent With these properties, $\mathcal{B}$ is formally given by
\vspace{-6pt}
\begin{align*}
\mathcal{B} \triangleq \bigcup_{i=1}^\infty \Big( R_i \! \oplus \! \big\{\mathbf{c}_i \big\}  \Big), \, \text{where} ~\, \mathbf{c}_i \! \in \Phi \,~ \text{and} ~\,  R_i \! \in \mathcal{R}. \\[-5ex] 
\end{align*}	
\end{definition}
\vspace{-2pt}

\begin{remark}
The sampled objects are considered to be placed with their center points at the origin before then being translated to their PPP center points, $\mathbf{c}_i$, via the Minkowski sum.
\end{remark}

\vspace{-9pt}
\begin{assumption} \label{NS_BS}
	Reflectors present within the network are randomly distributed according to the Boolean model $\mathcal{B}$.  For this model, we further assume: 
\vspace{-3pt}
\begin{itemize}
\item $\Phi$ has intensity $\lambda$;
\item $\mathcal{R}$ represents the set of all compact, non-empty, \emph{squares} in $\mathbb{R}^2$ with arbitrary orientation; and
\item For a given square in the sequence, $\{R_i\}_{i=1}^\infty$, we have $R_i \! \distras{\text{\scriptsize i.i.d.}}\! f_{W,\Theta}$, where $f_{W,\Theta}(w, \theta)$ is a bivariate, discrete distribution with $W$ representing the square's side-width and $\Theta$ its orientation.  The support of $f_{W, \Theta}$ is $\text{supp}\big(f_{W,\Theta}\big) = \setc[\Big]{ [w_i, \theta_j]^T \!\in \mathbb{R}^2 } {  \big\{w_i \big\}_{i = 1}^{n_w},~ \big\{\theta_j \big\}_{j=1}^{n_\theta} }$.\footnote{Our definition of $\mathcal{R}$ implies $W>0$ and is finite and that $0 \leq \Theta < \pi/2$. We further make one minor mathematical restriction in that $\Theta \neq 0$, \emph{i.e.}, now $0 < \Theta< \pi/2$.  We also assume $n_w >0$, $n_\theta >0$ are finite integers.}\textsuperscript{,}\footnote{The orientation is the angle, $\Theta = \theta$, between a vector emanating from the square's center (passing perpendicularly through an edge) and the $+x$-axis.}
\end{itemize} 
\end{assumption}

\vspace{-3pt}
	Next, since we aim to study a single arbitrary link, we have:
\vspace{-16pt}
\begin{definition}\label{TLSetup}
(\emph{Test Link (TL)}) The \emph{Test Link} is defined to be the link where the base station and mobile are separated by a distance $d>0$ and are at $\mathbf{b} = [-d/2, 0]^T$ and $\mathbf{m} = [d/2, 0]^T$, respectively, in a Cartesian coordinate system.
\end{definition}
\vspace{-7pt}
\begin{remark}
Under Assumption \ref{NS_BS}, the results derived in this paper apply, without loss of generality (w.l.o.g.), to any translation and/or orientation of the TL by the stationary and isotropic properties of the Boolean model.  See Fig. \ref{BM_Realization} for a realization of $\mathcal{B}$, under Assumption \ref{NS_BS}, over the TL setup. 
\end{remark}		
\vspace{-8pt}
\begin{assumption} \label{NoDiffraction}
Only first-order reflections are considered and the Specular Reflection Law holds at the reflector, \emph{i.e.}, the angle of incidence equals the angle of reflection. 
\end{assumption}
\vspace{-8pt}
\begin{remark}
This assumption stems from the fact that reflections tend to dominate NLOS propagation in mm-wave networks \cite{AndrewsHeath}.  Additionally, the effect of higher-order reflections is assumed to be minimal due to increased pathloss and reflection losses, \emph{e.g.}, \cite{Nor}, \cite{GDasConf}.  Although we place a particular emphasis on 5G networks, we note however that our model can roughly capture diffraction effects (although not ideal), since we can assume that a point of diffraction can be replaced by a properly positioned reflector.  Thus, our model is approximately applicable at lower frequencies as well, \emph{i.e.}, 3G/4G networks.\footnote{If at lower frequencies a more accurate model is introduced, which accounts for diffraction effects by placing diffractors according to a Boolean model, we would argue that the \emph{form} of our results would not change, \emph{i.e.}, the exponential \emph{form} of the distribution of the first-arriving NLOS path length should remain the same.  The reason why will become clear in Sec. \ref{NLoSdistribution}.}
\end{remark}
\vspace{-8pt}
\begin{assumption}
We assume that the base station and mobile are omni-directional, that is, either the base station and mobile are both equipped with isotropic antennas, or either the base station, mobile, or both are equipped with directional antennas that perform $360$ degree scans of the environment.  Thus, all possible reflection paths are illuminated.
\end{assumption}

\vspace{-10pt}
\begin{assumption} \label{No_Blocking}
Blocking effects on NLOS paths are ignored.  
\end{assumption}

\vspace{-7pt}
\begin{remark}
Although most works do include blocking, almost all, whether stated or not, assume \emph{independent} blocking, \emph{i.e.}, blocking on each NLOS path is independent.  However, in practice, \emph{dependent} (or \emph{correlated}) blocking is a large contributor, as it is quite reasonable that an object may block more than one reflected path at a time, especially when blockage size increases \cite{Aditya}.   
Consequently, this major assumption of independent blocking ultimately leads to results that are approximations.  Thus, in an effort to keep our mathematical treatment rigorous, we ignore blockages, which then provides a lower bound on the path length of the first-arriving NLOS signal.\footnote{Adding \emph{independent} blocking in our setting is trivial if an approximation of NLOS bias is desired, rather than a lower bound.} 
\end{remark}
\vspace{-2pt}

	Lastly, as we are studying NLOS bias, we implicitly assume that the LOS path is blocked throughout this work.

\vspace{-9pt}
\section{A Geometric Characterization of First-Order Reflection Positions}

\vspace{-2pt}
	Section \ref{Sec_Ref_Hyp} derives the \emph{reflection hyperbola}, \emph{i.e.}, the set of points where first-order reflections can occur for a particular reflector of fixed orientation but arbitrary position.  Then, this set is used to construct the \emph{reflection region}, which is described in detail in Section \ref{Ref_Region_Sec} and is the key to deriving the distribution of the first-arriving NLOS path length in Section \ref{NLoSdistribution}.  To facilitate these derivations, we begin with three important definitions.

\vspace{-4pt}
\begin{definition}
(\emph{Test Reflector (TR)}) The \emph{test reflector}, $r_{w, \theta}$, is defined to be a reflector with \emph{fixed} width $w$, \emph{fixed} orientation $\theta$, and \emph{arbitrary} center point $\mathbf{c} \in \mathbb{R}^2$.  
\end{definition}
\vspace{-8pt}
\begin{remark}
This definition allows us to choose a specific reflector and ``move it around'' $\mathbb{R}^2$ in search of potential first-order reflection points; which brings us to our next two definitions.
\end{remark}
\vspace{-7pt}

\begin{definition}
(\emph{Reflecting Edge}) Consider a TR $r_{w,\theta}$ under the TL setup from Definition \ref{TLSetup}.  We define the \emph{reflecting edge}, $e_i$ for $i \in \{\text{I}, \text{II}, \text{III}, \text{IV}\}$, to be the edge of $r_{w,\theta}$ that facilitates reflections in quadrant $i$, $Q_i$.  (See Fig. \ref{Reflecting_Edges}.)
\end{definition}	
\vspace{-2pt}
\begin{remark}
As a visual example, examining the reflections for the three different reflectors in Fig. \ref{BM_Realization}, we see that reflecting edge $e_\text{I}$ is responsible for first-order reflections between the base station and the mobile in $Q_\text{I}$, $e_\text{II}$ in $Q_\text{II}$, and so on.
\end{remark}
\vspace{-6pt}
\begin{definition}
(\emph{Potential Reflection Point (PRP)}) Consider a TR $r_{w, \theta}$.  A \emph{potential reflection point} is defined to be any point $\mathbf{x} \in \mathbb{R}^2$ such that an edge of $r_{w, \theta}$ can intersect $\mathbf{x}$ to produce a first-order reflection between the base station and mobile.  \emph{This implies the Specular Reflection Law holds at this point.}
\end{definition}
\vspace{-16pt}
\subsection{Derivation of the Reflection Hyperbola} \label{Sec_Ref_Hyp}
\vspace{-1pt}
	Using our new terminology, we wish to find the set of all the PRPs for TR $r_{w,\theta}$. This is given by the following lemma:

\begin{figure}[t]
\centering
\includegraphics [scale=0.30]{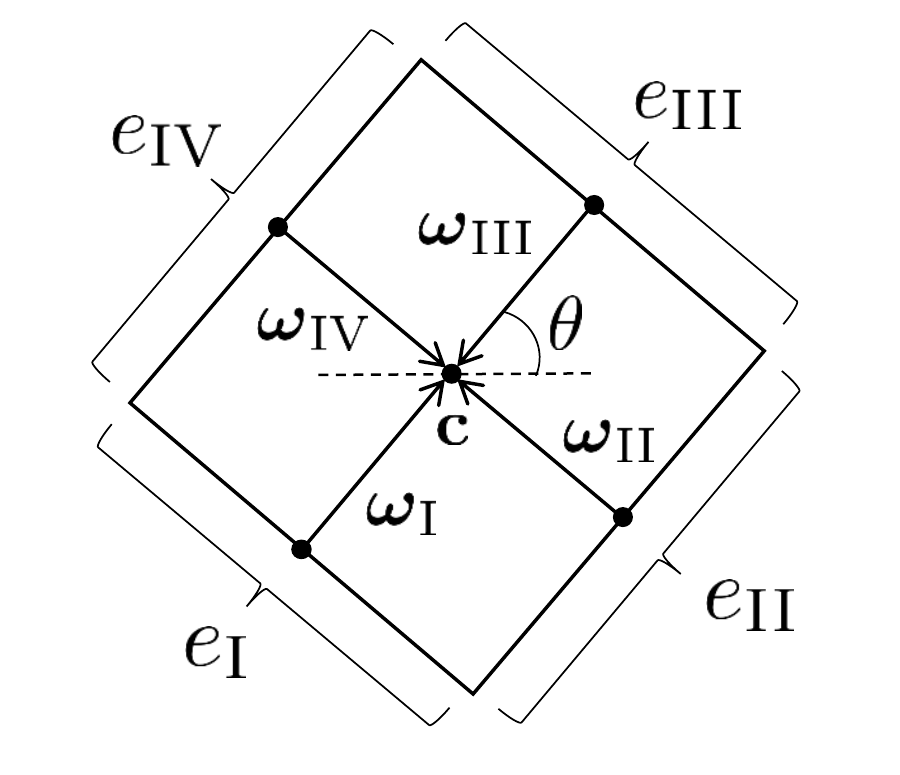}
\vspace{-7pt}
\caption{\textsc{Reflecting Edges for Test Reflector $r_{w,\theta}$.} For reference, the vector $\boldsymbol{\omega}_\text{I}$, from the edge's center point to the center of the reflector, always has angle $\theta$ (the reflector's orientation). The reflecting edge, $e_\text{I}$, is always opposite this vector.  The remaining `center-displacement vectors' and their corresponding edges are labeled in increasing order counterclockwise.\\[-5.0ex]}
\label{Reflecting_Edges}
\end{figure}

\begin{figure}[t]
\centering
\includegraphics [scale=0.32]{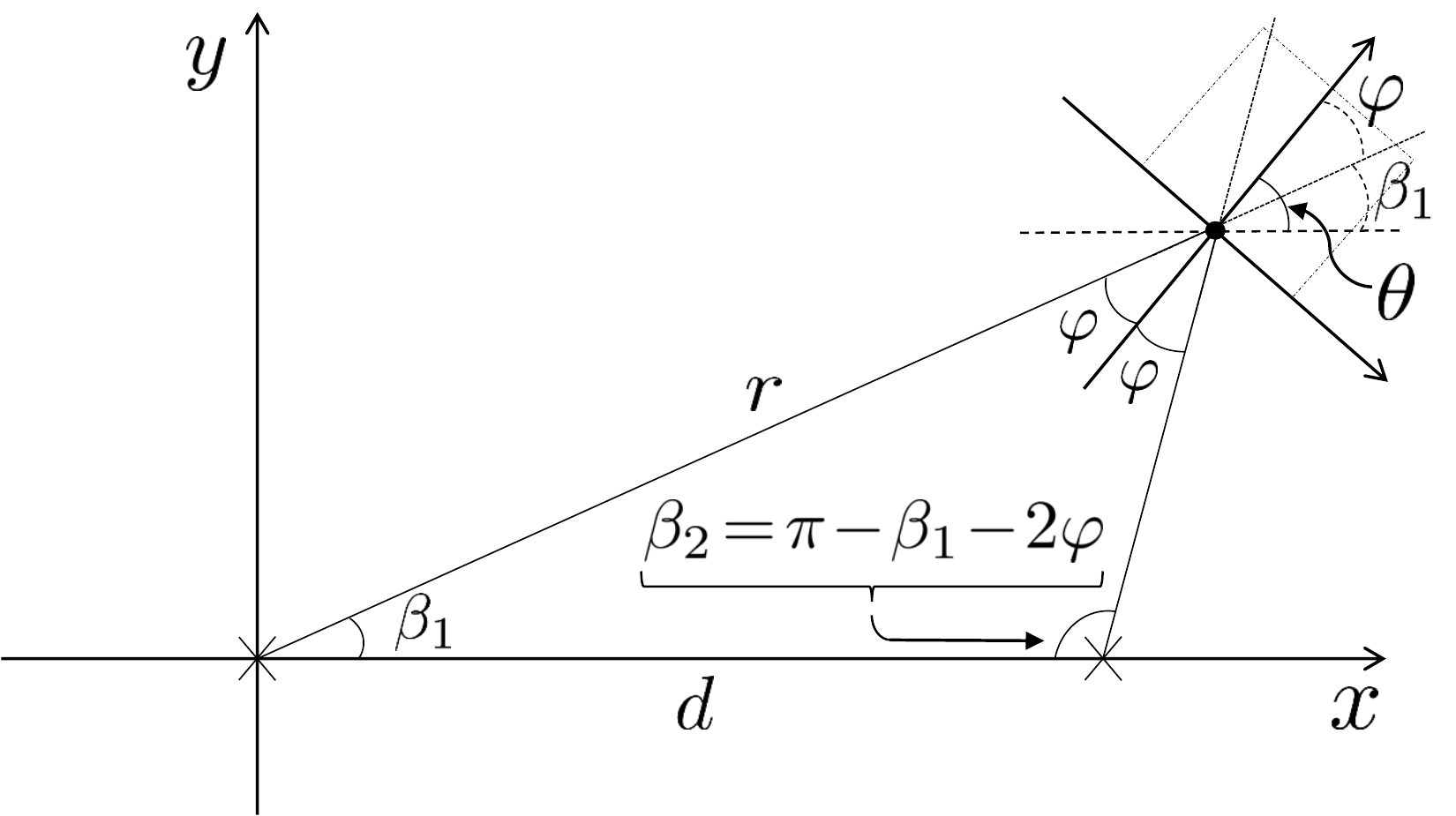}
\vspace{-4pt}
\caption{\textsc{Translation of the Test Link (TL).}  Depicted is a translation of the TL by $d/2$.  This allows for simple polar coordinate representations of the base station, mobile, and potential reflection points.\\[-5.5ex]}
\label{Polar_Coordinates}
\end{figure}

\vspace{-2pt}
\begin{lemma} \label{hyperbola}
Let $\theta \in (0, \pi/2)$ and consider the TL setup.  Then, the set of all PRPs for TR $r_{w,\theta}$, denoted by $\mathcal{H}_\theta$, is given by 
\begin{align} \label{Lemma_Equation}
\mathcal{H}_\theta = \setc[\Big]{ \! [x,y]^T \!\! \in  \mathbb{R}^2 } { y^2 - x^2 + 2\cot(2\theta)xy + d^2/4 = 0 \!}.
\end{align}
\end{lemma}

\begin{proof}
	To simplify the proof, we transform the TL setup by shifting the base station and mobile to the right by $d/2$ and then consider a polar coordinate system (see Fig. \ref{Polar_Coordinates}).  Let $\mathcal{H}_{\theta}^{'}$ be the set of PRPs in this transformed space.  Next, we convert the r.h.s. of (\ref{Lemma_Equation}) into this transformed space. Shifting to the right by $d/2$ yields:
\vspace{-3pt}
\begin{align*}
y^2 - (x-d/2)^2 + 2\cot(2\theta) (x-d/2) y + d^2/4 = 0. \\[-4ex]
\end{align*}
Using $\cot(2\theta) = \cos(2\theta)/\sin(2\theta)$, simplifying, and rearranging gives:
\vspace{-2pt}
\begin{align*}
\sin(2\theta)y^2 \! - \sin(2\theta)x^2 \!+\! 2\cos(2\theta)xy = d \! \cos(2\theta)y - d \! \sin(2\theta)x.\\[-4ex]
\end{align*}
Next, consider the polar coordinates $\ell \geq 0$ and $\alpha \in [0,2\pi)$ and the polar conversion equations $x = \ell \cos(\alpha)$ and $y = \ell \sin(\alpha)$.  Substituting these equations in for $x$ and $y$ above, after much simplification, yields the condition in the following set, which is the `translated-polar version' of (\ref{Lemma_Equation}):
\vspace{-2pt}
\begin{align} \label{Polar_Version}
\mathcal{H}_{\theta}^{'} = \setc[\Big]{ \! [\ell,\alpha]^T \!\! \in \mathbb{R}^2 } { \ell = d\sin \big( \alpha - 2\theta \big) \big/\sin \big( 2\alpha - 2\theta \big)\!}. \\[-4ex]\nonumber
\end{align}
W.l.o.g., establishing (\ref{Polar_Version}) will establish the claim. 
	
	\underline{\smash{($\subset$)}}:  Let $[r,\beta_1]^T \in \mathcal{H}_{\theta}^{'}$.  This implies $[r,\beta_1]^T$ is a PRP, hence the Specular Reflection Law must hold at this point.  From Fig. \ref{Polar_Coordinates}, $\varphi = \theta - \beta_1$ and $\beta_2 = \pi - \beta_1 - 2\varphi$ $\implies$ $\beta_2 = \pi + \beta_1 - 2\theta$.  Applying the Law of Sines gives:
$d \big/ \sin(2\varphi) = r \big/ \sin(\pi + \beta_1 - 2\theta)$, and substituting in for $\varphi$, simplifying, and solving for $r$ yields: $r = d \sin(\beta_1 - 2\theta) \big/ \sin (2\beta_1 - 2\theta)$ $\implies$ $[r,\beta_1]^T$ is in the r.h.s. of (\ref{Polar_Version}), as desired.

 \underline{\smash{($\supset$)}}: Reverse containment follows by applying the Law of Sines to the condition in the r.h.s. of (\ref{Polar_Version}) and then verifying that the angle of incidence equals the angle of reflection. 
\end{proof}

\vspace{-8pt}
\begin{remark}
Note the following: 1) $\mathcal{H}_{\theta} \neq \O$, \emph{i.e.}, there always exists a PRP for any reflector; 2) \emph{$r_{w,\theta}$ can produce a first-order reflection if and only if $r_{w,\theta}$ intersects $\mathcal{H}_\theta$} (with, of course, the correct reflecting edge); 3) the set condition in (\ref{Lemma_Equation}) is a hyperbola and thus we refer to the set of PRPs, $\mathcal{H}_\theta$, as the \emph{reflection hyperbola}; 4) the reflector orientation, $\theta$, is only present in the ``$xy$'' term of the hyperbola equation, which implies that changing the orientation of the reflector results in a rotation of this hyperbola about the origin; and 5) $\mathbf{b}$ and $\mathbf{m}$ are \emph{not} the foci of this hyperbola, but rather lie on the hyperbola itself.  See Fig. \ref{Reflection_Region} for an example of $\mathcal{H}_\theta$. 
\end{remark}

\vspace{-14pt}
\subsection{The Reflection Region and its Lebesgue Measure} \label{Ref_Region_Sec}
\vspace{-1pt}
	First, consider the TL setup, which we will work under for the remainder of this section.  Then, informally, we define the reflection region to be all of the points in $\mathbb{R}^2$ where TR $r_{w,\theta}$ may lie such that a first-order reflection between the base station and mobile is produced whose total path length is less than or equal to some distance, $s$, where $s>d$. This section is devoted to the formal mathematical construction of this region, along with the derivation of its Lebesgue measure.   
	
	First, as we wish to restrict our attention to first-order reflections of distance $\leq s$, we present the following definition:

\vspace{-5pt}
\begin{definition}\label{NLOSBE}
(\emph{Boundary Ellipse}) 
Let $\mathcal{E}_s$ be the set
\vspace{-2pt}
\begin{align*}
\mathcal{E}_s \triangleq \setc[\Big]{ (x,y) \in \mathbb{R}^2 } { x^2 \big/ u^2 + y^2 \big/ v^2 \leq 1},\\[-4ex]
\end{align*}
where $u^2 = s^2 / 4$ and $v^2 = (s^2 - d^2) / 4$.  We define the \emph{boundary ellipse} to be $\partial\mathcal{E}_s$, \emph{i.e.}, the set of all boundary points of $\mathcal{E}_s$.
\end{definition}	

\vspace{-6pt}
\begin{remark}
Note, $\partial\mathcal{E}_s$ forms an ellipse with foci at $\mathbf{b}$ and $\mathbf{m}$, see Fig. \ref{Reflection_Region}.  An important consequence of this definition is that if $r_{w,\theta}$ has PRPs in $\mathcal{E}_s$ then these PRPs are associated with first-order reflections of total reflected path length $\leq s$. 
\end{remark}
\vspace{-2pt}

	Next, since the PRPs of $r_{w,\theta}$ on the boundary, $\partial\mathcal{E}_s$, are associated with the NLOS path length $s$, \emph{i.e.}, the largest NLOS path length that is still within our reflection region, then these points will help define the outer edges of the reflection region.

\vspace{-4pt}
\begin{lemma} \label{Intersection_Point_Lemma}
For TR $r_{w,\theta}$, there are four PRPs in $\mathcal{H}_\theta \cap \partial\mathcal{E}_s$.  We denote these boundary intersection points by $\boldsymbol{\gamma}_i = [x_i, y_i]^T$, for $i \in \{ \text{I}, \text{II}, \text{III}, \text{IV} \}$, where the subscript corresponds to the quadrant in which the PRP resides. These are given by
\vspace{-4pt}
\begin{align*}
\boldsymbol{\gamma}_\text{\emph{I}} &= \bigg[ \!\sqrt{z_\text{\emph{I, III}}}~\!, \frac{v}{u} \sqrt{u^2 - \!z_\text{\emph{I,III}}} \bigg]^T \! \!, \! ~
\boldsymbol{\gamma}_\text{\emph{II}} =\bigg[ \!-\! \sqrt{z_\text{\emph{II, IV}}}~\!, \frac{v}{u} \sqrt{u^2 -\! z_\text{\emph{II,IV}}}\bigg]^T \! \! \!, \\[-4.2ex]
\end{align*}
$\boldsymbol{\gamma}_\text{\emph{III}} =-\boldsymbol{\gamma}_\text{\emph{I}}$, and $\boldsymbol{\gamma}_\text{\emph{IV}} =-\boldsymbol{\gamma}_\text{\emph{II}}$, where 
\vspace{-3pt}
\begin{align*}
z_\text{\emph{I, III}} =\frac{s^4 \cot^2\!\theta} {4 \Big[ s^2 \csc^2\!\theta - d^2 \Big] } ~~\, \text{and} ~~\, z_\text{\emph{II, IV}} = \frac{ s^4 \tan^2\!\theta} {4 \Big[ s^2 \sec^2\!\theta - d^2 \Big] },
\end{align*}
and $u$ and $v$ are from Def. \ref{NLOSBE}.  (See Fig. \ref{Reflection_Region} for $\boldsymbol{\gamma}_\text{\emph{I}}$ depiction.)
\end{lemma}
\vspace{-10pt}
\begin{proof}
Solve the system of two equations generated by $\mathcal{H}_\theta$ and $\partial\mathcal{E}_s$. (An alternate derivation of $\boldsymbol{\gamma}_\text{II}$ is given in \cite{Nor}.)
\end{proof}

\vspace{-6pt}

	Lastly, if we have a PRP for TR $r_{w, \theta}$ in $Q_i$, then to produce a first-order reflection at this PRP, any point along the reflecting edge $e_i$ will suffice.  Hence, we may finally construct the reflection region for TR $r_{w,\theta}$ and $s$ as follows:

\begin{figure}[t]
\centering
\includegraphics [scale=0.335]{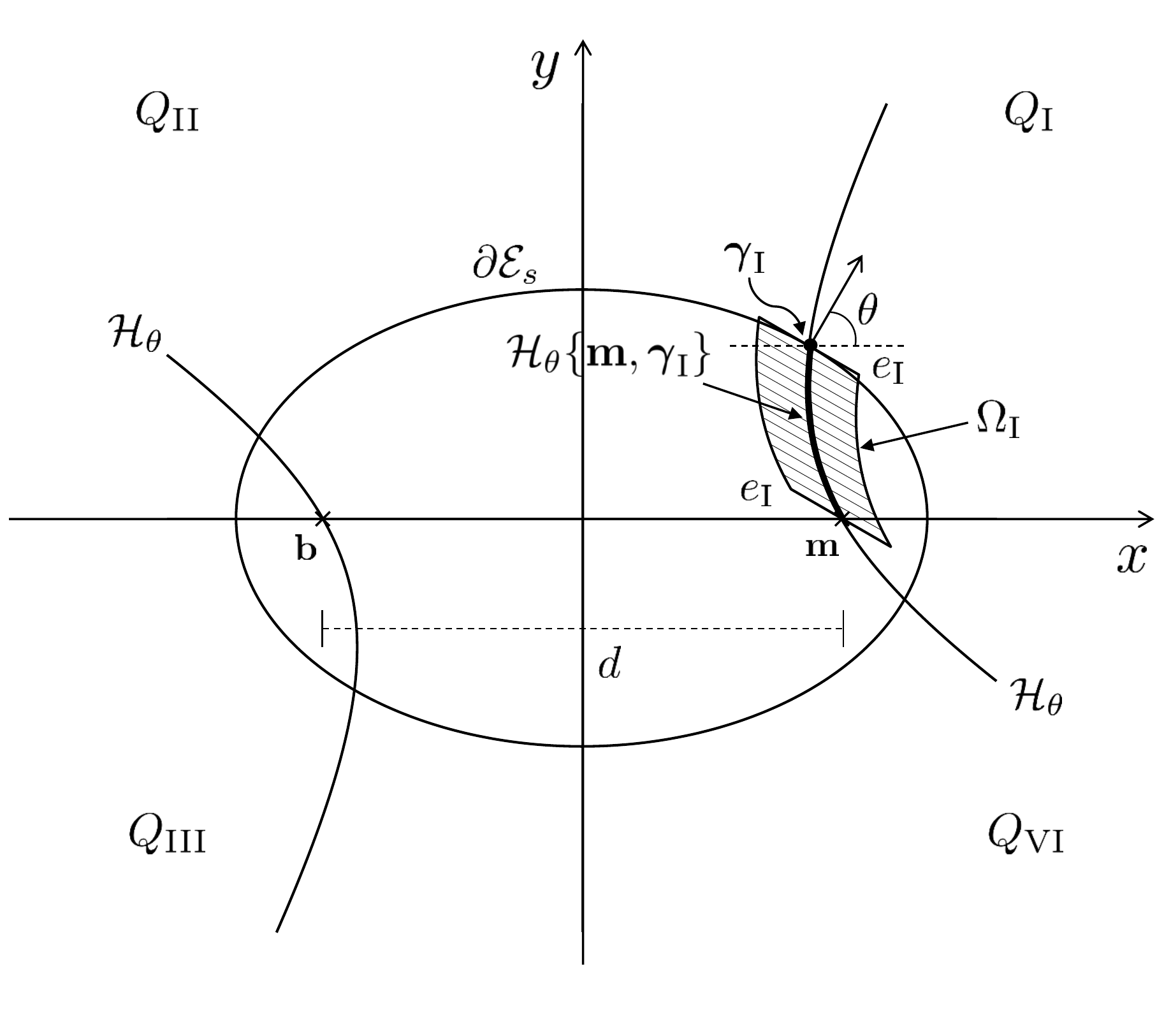}
\vspace{-15pt}
\caption{\textsc{The $Q_\text{I}$ Portion of The Reflection Region.} For test reflector $r_{w,\theta=\pi/3}$, the shaded region, $\Omega_\text{I}$, represents all of the points that \emph{the center of the reflecting edge, $e_\text{\emph{I}}$}, may lie in order to produce a first-order reflection of distance $\leq s$ meters.\\[-5.5ex]}
\label{Reflection_Region}
\end{figure}

\vspace{-4pt}
\begin{definition}\label{Ref_Reg}
(\emph{Reflection Region}) Consider $r_{w,\theta}$'s four reflecting edges, $e_\text{I}$ to $e_\text{IV}$, as four separate line segments and place the center of each one at the origin, preserving their respective orientations.  Further, let $\mathcal{H}_\theta \{\mathbf{p},\mathbf{q}\}$ denote the portion of $\mathcal{H}_\theta \cap \mathcal{E}_s$ between (and including) the points $\mathbf{p}$, $\mathbf{q}$ and define the sets
\vspace{-6pt}
\begin{align*} 
\Omega_\text{I} &\triangleq  \mathcal{H}_\theta \{\mathbf{m},\boldsymbol{\gamma}_\text{I}\} \oplus e_\text{I},  ~~~~~~ \! \Omega_\text{II} \triangleq  \mathcal{H}_\theta \{\mathbf{b},\boldsymbol{\gamma}_\text{II}\} \oplus e_\text{II}, \\
\Omega_\text{III} &\triangleq  \mathcal{H}_\theta \{\mathbf{b},\boldsymbol{\gamma}_\text{III}\} \oplus e_\text{III}, ~~~ \Omega_\text{IV} \triangleq  \mathcal{H}_\theta \{\mathbf{m},\boldsymbol{\gamma}_\text{IV}\} \oplus e_\text{IV}. \\[-4.5ex]
\end{align*}
Then, the \emph{reflection region} for $r_{w,\theta}$ and $s$ is defined as
\vspace{-6pt}
\begin{align}
\Omega \triangleq \bigcup_{i=\text{I}}^\text{IV} \Big( \Omega_i + \boldsymbol{\omega}_i \Big).
\end{align}
\end{definition}
\vspace{-6pt}
\begin{remark}
To clarify, $\Omega_\text{I} =  \mathcal{H}_\theta \{\mathbf{m},\boldsymbol{\gamma}_\text{I}\} \oplus e_\text{I}$, for example, is the region where the center of the reflecting edge, $e_\text{I}$, may lie in order to produce a first-order reflection.  Then, to have this region represent all of the points where $r_{w,\theta}$'s \emph{center} may lie, instead of $e_\text{I}$'s center, we simply need to shift this region by the displacement between $e_\text{I}$'s center and $\mathbf{c}$; hence, the addition of the $\boldsymbol{\omega}_i$s in the definition. For visual clarity, Fig. \ref{Reflection_Region} illustrates the $Q_\text{I}$ portion of the reflection region in terms of $e_\text{I}$'s center, not the reflector's center, $\mathbf{c}$.  That is, $\Omega_\text{I}$ is depicted, not $\Omega_\text{I} + \boldsymbol{\omega}_\text{I}$. 
\end{remark}
\vspace{-5pt}
\begin{remark}
The four quadrant portions of the reflection region intersect at most on a set of measure zero.  
\end{remark}
	
\vspace{-1pt}
	In order to calculate the Lebesgue measure of the reflection region, we must calculate Lebesgue measures of regions generated by Minkowski sums; a difficult task in general.  However, in our specific case here, the next lemma reveals that this is relatively straightforward.
\vspace{-5pt}	
\begin{lemma} \label{Integration_Lemma}
Let $a,b \in \mathbb{R}$, where $a<b$, and let $f:[a,b] \rightarrow \mathbb{R}$ be a $\mu_1$-measurable function.\footnote{$\mu_n$ represents the $n$-dimensional Lebesgue measure.}  Further, let $\mathcal{L}_h$ be a vertical line segment that begins at $[0, -h/2]^T$ and ends at $[0, h/2]^T$. Then, $\mu_2 \big( f \oplus \mathcal{L}_h \big) = h(b-a)$. (See Fig. \ref{Integration_Lemma_Fig} for an illustration.) 
\end{lemma}

\begin{figure}
\centering
\begin{tikzpicture} [scale=0.38]

      \fill [pattern=north west lines, domain=-2:3, variable=\x]
      (-2, 0)
      -- plot ({\x}, {sqrt(2*\x+4)+1.5})
      -- (3, 0)
      -- cycle;
      
      \fill [white, domain=-2:3, variable=\x]
      (-2, 0)
      -- plot ({\x}, {sqrt(2*\x+4)+0.5})
      -- (3, 0)
      -- cycle;

      \draw[->] (-5,0) -- (5,0) node[below] {$x$};
      \draw[->] (0,-0.5) -- (0,5) node[left] {$y$};
      
      \draw[line width=0.1mm,scale=1,domain=-2:3,smooth,variable=\x]  plot ({\x},{sqrt(2*\x+4)+1});
      
      \draw[line width=0.3mm,scale=1,domain=-2:3,smooth,variable=\x]  plot ({\x},{sqrt(2*\x+4)+1.5});
      \draw[line width=0.3mm,scale=1,domain=-2:3,smooth,variable=\x]  plot ({\x},{sqrt(2*\x+4)+0.5});
      
      \draw[line width=0.3mm] (-2,0.5) -- (-2,1.5);
      \draw[line width=0.4mm] (3,4.162+0.5) -- (3,4.162-0.5);
      
      \draw[dashed] (-2,0.5) -- (-2,0) node[below] {$a$};
      \draw[dashed] (3,4.162-0.5) -- (3,0) node[below] {$b$};
      
      \draw[-]   (3.3,4.162+0.5) -- (3.3,4.162-0.5);
      \draw[-]   (3.2,4.162-0.5) -- (3.4,4.162-0.5);
      \draw[-]   (3.2,4.162+0.5) -- (3.4,4.162+0.5);
      
      \node[] at (3.7,4.162)  {$h$};

      \draw[-]   (-2.3,0.5) -- (-2.3,1.5);
      \draw[-]   (-2.2,0.5) -- (-2.4,0.5);
      \draw[-]   (-2.2,1.5) -- (-2.4,1.5);
      
      \node[] at (-2.7,1)  {$h$};

      \node[] at (1.6,2.2)  {\footnotesize$f(x) - \frac{h}{2}$};
      \node[] at (1.65,5.2)  {\footnotesize$f(x) + \frac{h}{2}$};
      
      \draw[<-] (-1,1.414+1.6) -- (-1.9,3.7) node[left] {\footnotesize$ f \oplus \mathcal{L}_h$};

\end{tikzpicture}
\vspace{-17pt}
\caption{\textsc{Illustration of Lemma \ref{Integration_Lemma}.\\[-7ex]}  
}
\label{Integration_Lemma_Fig}
\end{figure}
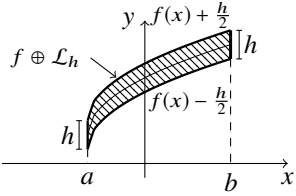

\vspace{-10pt}
\begin{proof}
We have the following
\vspace{-5pt}
\begin{align*}
\mu_2  \big( f  \oplus \mathcal{L}_h \big) &= \int_a^b \Big[ f(x) + h/2 \Big] - \Big[f(x) - h/2 \Big] \, \text{d}  \mu_1 \\
&= \int_a^b \!\! h \, \, \text{d}\mu_1  \, \, = \, \, h \mu_1 \big( [a,b] \big) \, \, = \, \, h(b-a).  \qedhere
\end{align*}
\end{proof}

\vspace{-9pt}
\begin{theorem} \label{LM_Ref_Reg}
The Lebesgue measure of the reflection region is given by
\vspace{-7pt}
\begin{align*}
\mu_2\big( \Omega \big) \! = w \Big[ \! \sqrt{s^2 \!-\! d^2 \!\sin^2\!\theta} - d(\sin \theta + \cos \theta) +\! \sqrt{s^2 \!-\! d^2\! \cos^2\!\theta} \, \Big]. 
\end{align*}
\end{theorem}

\vspace{-14pt}
\begin{proof}
We begin by noting that
\vspace{-4pt}
\begin{align}\label{four}
\mu_2 \big( \Omega \big) &\stackrel{(a)}{=} \mu_2 \Bigg( \bigcup_{i=\text{I}}^\text{IV} \Big( \Omega_i + \boldsymbol{\omega}_i \Big) \Bigg) \,\,\stackrel{(b)}{=} \,\, \sum_{i=\text{I}}^\text{IV} \mu_2 \Big( \Omega_i + \boldsymbol{\omega}_i \Big)  \nonumber\\
&\stackrel{(c)}{=} \sum_{i=\text{I}}^\text{IV} \mu_2 \big( \Omega_i \big) \,\, \stackrel{(d)}{=}\,\,  2 \mu_2 \big( \Omega_\text{I} \big) + 2 \mu_2 \big( \Omega_\text{IV} \big),\\[-5ex] \nonumber
\end{align}
where $(a)$ follows from Definition \ref{Ref_Reg}, $(b)$ from finite additivity, $(c)$ from translation invariance, and $(d)$ by the symmetry of $\Omega_\text{I}$ with $\Omega_\text{III}$ and of $\Omega_\text{II}$ with $\Omega_\text{IV}$.

	Hence, we must find $\mu_2 \big( \Omega_\text{I} \big)$ and $\mu_2 \big( \Omega_\text{IV} \big)$.  We start with $\mu_2 \big( \Omega_\text{I} \big)$.  First, consider rotating the coordinate system counterclockwise by $\theta$ (see Fig. \ref{Reflection_Region} for a visual -- rotate the coordinate axis while keeping the reflection region stationary).  Next, in this new coordinate system, $e_\text{I}$ is now a vertical line segment of height $w$ and the set $\mathcal{H}_\theta \{\mathbf{m},\boldsymbol{\gamma}_\text{I}\}$ can now be regarded as a $\mu_1$-measurable function on $\big[x_\theta^m, x_\theta^{\gamma_\text{I}}\big]$, where $x_\theta^m$ and $x_\theta^{\boldsymbol{\gamma}_\text{I}}$ are the corresponding $x$-coordinates of $\mathbf{m}$ and $\boldsymbol{\gamma}_\text{I}$ in this rotated coordinate system. They are given by $x_\theta^m \!= \! (d/2) \cos \theta$ and $x_\theta^{\boldsymbol{\gamma}_\text{I}} \!=\! x_\text{I} \cos \theta + y_\text{I} \sin \theta$, where these were obtained by multiplying $\mathbf{m}$ and $\boldsymbol{\gamma}_\text{I}$ by the rotation matrix $\Bigl[\begin{smallmatrix}\cos \theta & \sin \theta \\ -\sin \theta & \cos \theta \end{smallmatrix}\Bigr]$.  Hence, we may now use Lemma \ref{Integration_Lemma} to compute $\mu_2(\Omega_\text{I})$:
\vspace{-5pt}
\begin{align}
\mu_2 \big( \Omega_\text{I} \big) &\stackrel{(a)}{=} \mu_2 \big( \mathcal{H}_\theta \{\mathbf{m},\boldsymbol{\gamma}_\text{I}\} \oplus e_\text{I} \big)
\stackrel{(b)}{=} w \Big[ x_\theta^{\boldsymbol{\gamma}_\text{I}} - x_\theta^m \Big] \nonumber\\
&\stackrel{(c)}{=} \frac{w}{2} \bigg[\sin \theta \sqrt{s^2 \csc^2\theta - d^2} - d\cos \theta \bigg] \label{result_simplification}, \\[-4ex] \nonumber
\end{align}
where $(a)$ follows by definition, $(b)$ by the use of Lemma \ref{Integration_Lemma} in the rotated coordinated system, and $(c)$ by substituting in for $x_\theta^m$ and $x_\theta^{\boldsymbol{\gamma}_\text{I}}$, using Lemma \ref{Intersection_Point_Lemma} for $x_\text{I}$ and $y_\text{I}$, and simplifying.
	
	Note, $\mu_2 \big( \Omega_\text{IV} \big)$ can be obtained similarly and yields:
\vspace{-4pt}
\begin{align}
\mu_2 \big( \Omega_\text{IV} \big) = \frac{w}{2} \bigg[\cos \theta \sqrt{s^2 \sec^2\theta - d^2} - d\sin \theta \bigg] \label{Omega_4}.
\end{align}

	Substituting (\ref{result_simplification}) and (\ref{Omega_4}) into (\ref{four}) establishes the claim.
\end{proof}
\vspace{-9pt}
\begin{remark}
\emph{Moving forward, we use the notation $\Omega(w, \theta, s)$ to represent the reflection region generated by TR $r_{w, \theta}$ and maximum (first-order reflection) NLOS path length $s$.}\footnote{The reflection region is also implicitly a function of $d$, however, since $d$ remains constant for a given TL setup, we omit it from this notation.}
\end{remark}


\vspace{-9pt}
\section{The Distribution of the Path Length of the First-Arriving NLOS Signal}\label{NLoSdistribution}
\vspace{-1pt}
\begin{theorem} \label{main_thrm}
Let $S$ be the random variable representing the length, in meters, of the shortest (i.e., first-arriving) NLOS path under the TL setup.  Then, under the assumptions in Sec. \ref{Model_Assumptions}, the CDF of $S$ is
\vspace{-6pt}
\begin{align} \label{Dist_of_S}
F_S(s) = 1 - e^{-\lambda \sum\nolimits_{i=1}^{n_w} \sum\nolimits_{j=1}^{n_\theta} f_{W,\Theta}(w_i, \theta_j \!) \,\, \mu_2 \big( \Omega(w_i, \theta_j, s) \big) }, \\[-4.5ex]\nonumber
\end{align}
where $\lambda, n_w, n_\theta$, and $f_{W, \Theta}(w, \theta)$ are given in Assumption \ref{NS_BS}, $\mu_2 \big( \Omega (w, \theta, s) \big)$ is given in Theorem \ref{LM_Ref_Reg}, and $\text{\emph{supp}}(S) = (\text{d}, \infty)$.
\end{theorem}
\vspace{-12pt}
\begin{proof}
We have: $F_S(s) = P[S \leq s] = 1 - P[S>s]$
\vspace{-4pt}
\begin{align*}
&\stackrel{(a)}{=} ~ 1 - \prod_{i=1}^{n_w} \prod_{j=1}^{n_\theta} P\big[\text{No reflectors fall within } \Omega (w_i, \theta_j, s) \big] \\
&\stackrel{(b)}{=} ~ 1 - \prod_{i=1}^{n_w} \prod_{j=1}^{n_\theta} e^{- \lambda \, f_{W,\Theta}(w_i, \theta_j) \, \mu_2 \big( \Omega (w_i, \theta_j, s) \big)}, \\[-5ex]
\end{align*}
where $(a)$ follows from the fact that for each $w_i$ and $\theta_j$ we have independently thinned PPPs and $(b)$ follows from the void probabilities of these thinned PPPs, with thinning density $\lambda f_{W,\Theta}(w_i, \theta_j)$.  Simplifying yields (\ref{Dist_of_S}), as desired.
\end{proof}
\vspace{-9pt}
\begin{remark}
Note that the shortest NLOS path length, $s$, appears in the argument of the exponentials in $(b)$ above.  \emph{Thus, the exponential form of the shortest NLOS path length distribution arises from the void probability of a PPP.}
\end{remark}

\begin{remark}
\emph{Let $B$ be the random variable denoting the NLOS bias distance.  Then, the distribution of B is obtained through the simple relationship: $B = S - d$.}
\end{remark}

	Although not obvious, the argument of the exponent in (\ref{Dist_of_S}) is nearly linear in $s$, and thus, close to a true exponential distribution.  We now digress from the rigor to give a true exponential distribution approximation to (\ref{Dist_of_S}).
\vspace{-4pt}
\begin{approximation} \label{Approx}
Under the same assumptions as Theorem \ref{main_thrm}, the distribution of the shortest NLOS path, $S$, can be approximated by the exponential distribution 
\vspace{-5pt}
\begin{align} \label{approx_formula}
F_S(s) \approx 1 - e^{-2 \, \lambda \, \mathop{{}\mathbb{E}}[W]  \, (s - d)}.
\end{align}
\end{approximation}
\vspace{-15pt}
\begin{proof}
We present a heuristic derivation based on asymptotics.  First, if $X$ is an exponential random variable, then in its distribution, the exponential argument is linear in $x$.  Thus, if we set $g(s)$ equal to the exponential argument in (\ref{Dist_of_S}), then we wish to find a linear approx. $h(s)=ps+q$, where $h(s) \approx g(s)$.

	To find $p$, consider the slope of $g(s)$ for $s \gg d$, \emph{i.e.},
\vspace{-6pt}
\begin{align*}
p &\approx\!  \lim_{s\to \infty} \frac{\text{d}}{\text{d}s} \big[g(s)\big] \\ 
&= \!\lim_{s\to \infty} \!- \lambda \! \sum\limits_{i=1}^{n_w} \! \sum\limits_{j=1}^{n_\theta}  f_{W,\Theta}(w_i, \!\theta_j \!) w_i  \textstyle\Bigg[ \!\frac{s}{\sqrt{s^2 - d^2\! \sin^2\!\theta_j}} \!  +  \! \frac{s}{\sqrt{s^2 - d^2\! \cos^2\!\theta_j}}\!\Bigg] \\
& = -2 \lambda \! \sum\limits_{i=1}^{n_w} \! \sum\limits_{j=1}^{n_\theta}  f_{W,\Theta}(w_i, \!\theta_j \!) w_i \,=\, -2 \lambda \mathop{{}\mathbb{E}}[W]. \\[-5ex]
\end{align*}
Since we desire our approximation to have the same support as the original, we set $g(d) = 0 = h(d)\!\! \implies\!\! q = 2 \lambda \mathop{{}\mathbb{E}}[W]d$.  Thus, $h(s) = -2 \lambda \mathop{{}\mathbb{E}}[W] (s-d)$, and Approximation \ref{Approx} follows. \qedhere
\end{proof}
	  
\vspace{-10pt}  
\begin{remark}
\emph{Not only is this approximation a true exponential distribution, but it also reveals that the distribution of the first-arriving NLOS path length (i.e., the NLOS bias) depends on the density of the reflectors along with their average size.}
\end{remark}

\vspace{-12pt}
\section{Numerical Results} \label{Numerical_Results}
\vspace{-1pt}
	This section compares the theoretical NLOS bias distribution, Theorem \ref{main_thrm}, along with its exponential distribution approximation, Approx. \ref{Approx}, against that obtained from a simulated network of reflectors/buildings and against a common exponential NLOS bias model from the literature \cite{Chen}.  
	
	From Fig. \ref{cdfs_NLoS_Bias}, we can see that Theorem \ref{main_thrm} matches that of simulation, verifying the derivation.  Further, Approx. \ref{Approx} also provides a close match to Theorem \ref{main_thrm}.  Examining the exponential bias model in \cite{Chen}, which was generated, in part, from measurement data at 900MHz, we see that our Approx. \ref{Approx} provides a close fit, despite our model being an approximation at lower frequencies (Assumption \ref{NoDiffraction}).  Note that the close fit is partially due to the rural setting chosen, since this minimizes the effect of blockages, which we do not consider (Assumption \ref{No_Blocking}).  As the density of blockages increases, \emph{e.g.}, to suburban/urban settings, Assumption \ref{No_Blocking} becomes less applicable, and consequently, our model offers a lower bound on the NLOS bias in these scenarios.  Lastly, in addition to our distribution providing an analytical backing for older exponential bias models at lower frequencies, we postulate that our model will offer an even better characterization of NLOS bias in 5G networks.  This stems from the fact that our results were derived under the assumption that NLOS propagation is predominately due to reflections -- a defining feature of 5G mm-wave channels.

\vspace{-10pt}
\section{Conclusion}
\vspace{-2pt}

	Under first-order reflections and a setup of reflectors with random orientations, sizes, and placements, this paper set out to analytically explore the bias experienced on an NLOS range measurement for a typical anchor-target link.  Out of the subsequent analysis arose an NLOS bias distribution that exhibited an exponential form and that could be closely approximated by an exponential distribution.  This result is not only consistent with the exponential model for NLOS bias seen in the localization literature, \emph{e.g.} \cite{Swaroop}, \cite{exp_justification}, but it was further shown to closely match a commonly used exponential model in environments with low blockage density \cite{Chen}.  As blockage density increases, our distribution provides a lower bound on the NLOS bias.  Finally, although further analysis is warranted, these initial results suggest that out of the many range measurement NLOS bias models that exist in the localization literature, an exponentially distributed model may now be the first with a rigorous analytical backing.

\begin{figure}[t]
\centering
\includegraphics [scale=0.41]{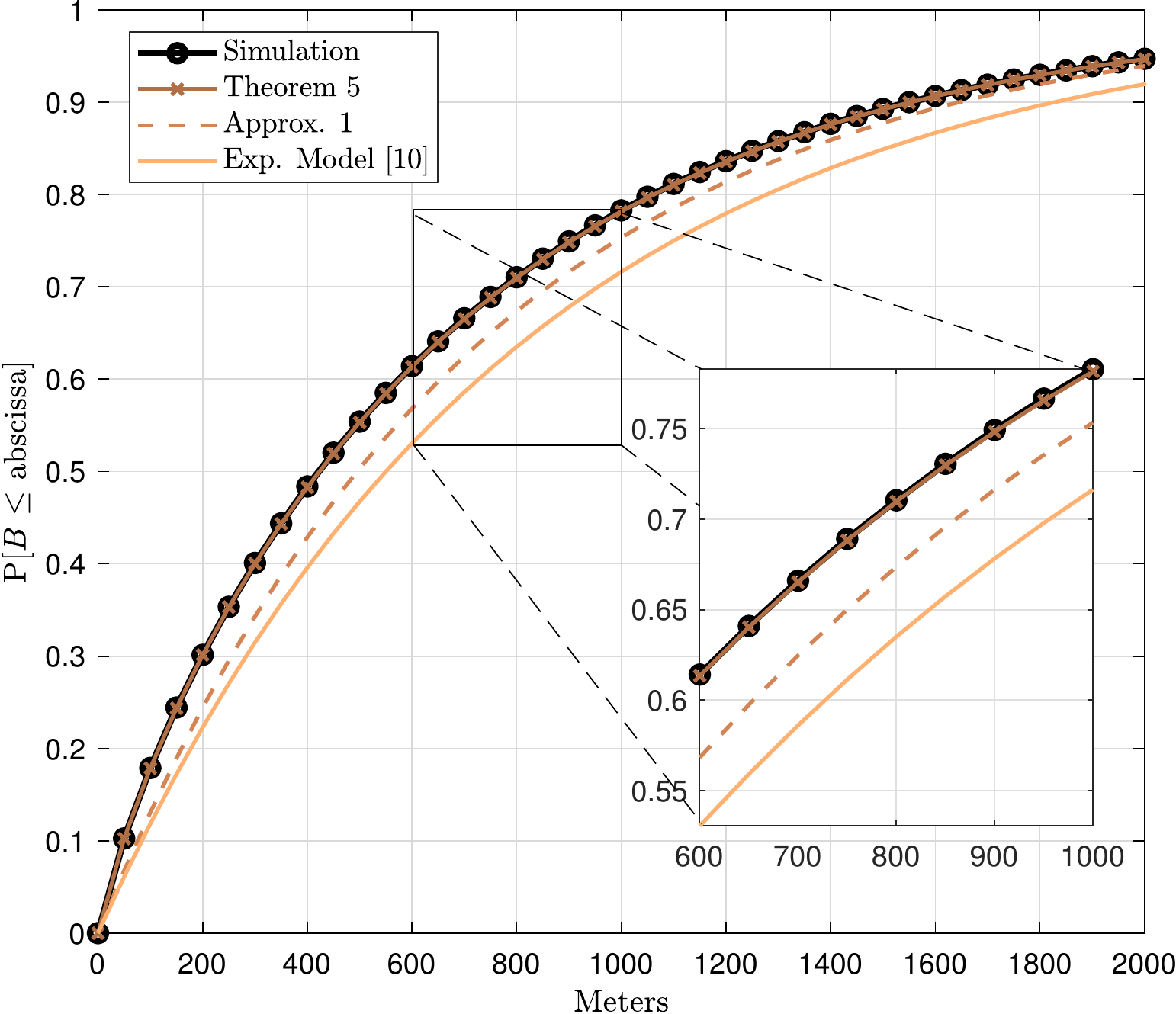}
\vspace{-4pt}
\caption{\textsc{NLOS Bias CDFs}. All four results were conducted on the TL setup of $d = 300m$ and by noting that $B = S-d$.  This comparison was done in a rural setting. Thus, the parameters for Simulation, Thrm. \ref{main_thrm}, and, Approx. \ref{Approx} were chosen s.t. there were an avg. of 10 buildings per \emph{sq. km} (used to set $\lambda$) and the buildings were chosen from the joint uniform distribution where $\text{supp}(W) = \{20m, 40m, \dots, 120m\}$, $\text{supp}(\Theta) = \{ 10^\circ, 20^\circ, \dots, 80^\circ \}$, and $f_{W, \Theta}(w, \theta) = 1/(6 \!\cdot\! 8)$ (this implies $\mathop{{}\mathbb{E}}[W] = 70m$ and that an avg. of $5\%$ of the land was covered by buildings).  Similarly, for the model in \cite[eqs. (11) and (12)]{Chen}, the parameters were chosen according to the rural settings from \cite[Table I]{Chen}, with $\xi$ set equal to it's mean.  Lastly, the Simulation CDF was generated over $10^6$ Boolean model realizations of buildings on an $8km \times 8km$ grid and by finding the shortest NLOS path in each realization. \\[-5.6ex]}
\label{cdfs_NLoS_Bias}
\end{figure}


\end{document}